\documentclass{llncs}

\usepackage{latexsym,amsfonts,amsmath,amssymb,amscd, graphicx}

\renewcommand{\epsilon}{\varepsilon}

\newcommand{\dc}{d_C}

\newcommand{\N}{\mathbb{N}}

\newcommand{\Z}{\mathbb{Z}}
\newcommand{\zz}{{\mathbb{Z}^2}}

\newcommand{\M}{\mathbb{M}}
\newcommand{\U}{\mathbb{U}}

\newcommand{\A}{\mathcal{A}}

\newcommand{\azz}{\mathcal{A}^{\mathbb{Z}^2}}

\newcommand{\am}{\mathcal{A}^{\mathbb{M}}}

\newcommand{\s}{\sigma}

\newcommand{\dd}{\delta}

\newcommand{\e}{\epsilon}

\newcommand{\gs}{\Sigma}

\newcommand{\acf}{F}
\newcommand{\obst}{\Sigma_{\acf}}
\newcommand{\obstset}{\A_F}
\newcommand{\oT}{\downarrow}
\newcommand{\oB}{\uparrow}
\newcommand{\oL}{\rightarrow}
\newcommand{\oR}{\leftarrow}
\newcommand{\oTL}{\searrow}
\newcommand{\oTR}{\swarrow}
\newcommand{\oBL}{\nearrow}
\newcommand{\oBR}{\nwarrow}
\newcommand{\roto}{\text{$\curvearrowright$}}
\newcommand{\rotb}{\text{\raisebox{1ex}{\rotatebox{180}{$\curvearrowright$}}}}
\newcommand{\slot}[1]{\text{\hbox to .8cm{\hfill\vbox to.2cm{\vfill \hbox{$#1$}\vfill}\hfill}}}
\newcommand{\minislot}[1]{\text{\hbox to.4cm{\hfill\vbox to.2cm{\vfill \hbox{$#1$}\vfill}\hfill}}}
\newcommand{\posout}{\blacksquare}
\newcommand{\posin}{ }
\newcommand{\nslot}[1]{{\hbox to.4cm{\hfill\vbox to.2cm{\vfill \hbox{$#1$}\vfill}\hfill}}}

\newcommand{\equpt}{\mathcal{E}_q}
\newcommand{\sensi}{\mathcal{S}}
\newcommand{\nono}{\mathcal{N}}
\newcommand{\phivarouch}{\Phi_1}
\newcommand{\phiclassdir}{\Phi_2}
\newcommand{\phiclassinv}{\Phi_3}
\newcommand{\phiemboit}{\Phi_4}
\newcommand{\machine}[1]{\mathcal{M}_{#1}}

\newcommand{\obstacl}[1]{\Sigma_{#1}}
\newcommand{\tileset}[1]{T_{#1}}

\newcommand{\outo}{-}

\newcommand{\transirot}[9]{
  \begin{array}[c]{c}
    \roto{}\\
    \text{$\begin{array}{c|c|c}
          \slot{#1} & \slot{#2} & \slot{#3}\\
          \hline
          \slot{#4} & \slot{#5} & \slot{#6}\\
          \hline
          \slot{#7} & \slot{#8} & \slot{#9}
        \end{array}$}\\
    \rotb{}
  \end{array}\mapsto 
}
\newcommand{\transi}[9]{
    \text{$\begin{array}{c|c|c}
          \slot{#1} & \slot{#2} & \slot{#3}\\
          \hline
          \slot{#4} & \slot{#5} & \slot{#6}\\
          \hline
          \slot{#7} & \slot{#8} & \slot{#9}
        \end{array}$}\mapsto 
}

\title{Topological Dynamics of 2D Cellular Automata}

\author{Mathieu Sablik\thanks{\email{mathieu.sablik@umpa.ens-lyon.fr}, \email{sablik@cmi.univ-mrs.fr}}\inst{1}
\and Guillaume Theyssier\thanks{\email{guillaume.theyssier@univ-savoie.fr}}\inst{2}}

\institute{UMPA, (UMR
  5669 --- CNRS, ENS Lyon), 46, all\'ee d'Italie 69364 Lyon cedex 07
  FRANCE\\
  LATP, (UMR 6632 --- CNRS, Universit\'e de Provence), CMI, 
  Universit\'e de Provence, Technop\^ole Ch\^ateau-Gombert, 39, rue F. Joliot Curie, 13453 Marseille Cedex 13 FRANCE 
  \and LAMA, (UMR 5127 --- CNRS, Universit\'e de Savoie), Campus
  Scientifique, 73376 Le Bourget-du-lac cedex FRANCE}
\begin{document}

\maketitle

\begin{abstract}
  Topological dynamics of cellular automata (CA), inherited from
  classical dynamical systems theory, has been essentially studied in
  dimension 1. This paper focuses on 2D CA and aims at showing that
  the situation is different and more complex. The main results are
  the existence of non sensitive CA without equicontinuous points, the
  non-recursivity of sensitivity constants and the existence of CA
  having only non-recursive equicontinuous points. They all show a
  difference between the 1D and the 2D case.  Thanks to these new
  constructions, we also extend undecidability results concerning
  topological classification previously obtained in the 1D case.
\end{abstract}

\section{Introduction}

Cellular automata were introduced by J. von Neumann as a simple formal
model of cellular growth and replication. They consist in a discrete
lattice of finite-state machines, called {\em cells}, which evolve
uniformly and synchronously according to a local rule depending only
on a finite number of neighboring cells. A snapshot of the states of
the cells at some time of the evolution is called a {\em
  configuration}, and a cellular automaton can be view as a global
action on the set of configurations.

Despite the apparent simplicity of their definition, cellular automata
can have very complex behaviours. One way to try to understand this
complexity is to endow the space of configurations with a topology and
consider cellular automata as classical dynamical systems. With such a
point of view, one can use well-tried tools from dynamical system
theory like the notion of sensitivity to initial condition or the
notion of equicontinuous point.

This approach has been followed essentially in the case of
one-dimensional cellular automata. P.  K{\accent"17 u}rka has shown
in~\cite{Kurka97} that 1D cellular automata are partitioned into two
classes:
\begin{itemize}
\item $\equpt$, the set of cellular automata with equicontinuous points,
\item $\sensi$, the set of sensitive cellular automata.
\end{itemize}
We stress that this partition result is false in general for classical
(continuous) dynamical systems. Thus, it is natural to ask whether
this result holds for the model of CA in any dimension, or if it is a
``miracle'' or an ``anomaly'' of the one-dimensional case due to the
strong constraints on information propagation in this particular
setting.  One of the main contributions of this paper is to show that
this is an anomaly of the 1D case (section~\ref{sec:sensitivity}):
there exist a class $\nono$ of 2D CA which are neither in $\equpt$ nor
in $\sensi$.

Each of the sets $\equpt$ and $\sensi$ has an extremal sub-class:
equicontinous and expansive cellular automata (respectively).  This
allows to classify cellular automata in four classes according to the
degree of sensitivity to initial conditions. The dynamical properties
involved in this classification have been intensively studied in the
literature for 1D cellular automata (see for
instance~\cite{Kurka97,BlanchardMaass,tisseur,permExp}).
Moreover, in~\cite{varouch}, the undecidability of this classification
is proven, except for the expansivity class whose decidability remains
an open problem.

In this paper, we focus on 2D CA and we are particularly interested in
differences from the 1D case. As said above, we will prove in
section~\ref{sec:sensitivity} that there is a fundamental difference
with respect to the topological dynamics classification, but we will
also adopt a computational complexity point of view and show that some
properties or parameters which are computable in 1D are non recursive
in 2D (proposition~\ref{prop:constant} and \ref{prop:nonrecpoint} of
section~\ref{sec:classif}). To our knowledge, only few
dimension-sensitive undecidability results are known for CA
(\cite{kari94,Bernardi}).  However, we believe that such subtle
differences are of great importance in a field where the common belief
is that everything interesting is undecidable.

Moreover, we establish in section~\ref{sec:classif} several complexity
lower bounds on the classes defined above and extend the
undecidability result of~\cite{varouch} to dimension 2. Notably, we
show that each of the class $\equpt$, $\sensi$ and $\nono$ is neither
recursively enumerable, nor co-recursively enumerable. This gives new
examples of ``natural'' properties of CA that are harder than the
classical problems like reversibility, surjectivity or nilpotency
(which are all r.e. or co-r.e.).

\section{Definitions}

Let $\A$ be a finite set and $\M=\Z$ (for the one-dimensional case) or
$\Z^2$ (for the two-dimensional case). We consider $\am$, the {\em
  configuration space} of $\M$-indexed sequences in $\A$. If $\A$ is
endowed with the discrete topology, $\am$ is compact, perfect and
totally disconnected in the product topology.  Moreover one can define
a metric on $\am$ compatible with this topology:
$$\forall x,y\in\am,\quad \dc(x,y)=2^{-\min\{\|i\|_\infty: x_i\ne y_i \ i\in\M \}}.$$

Let $\U\subset\M$. For $x\in\am$, denote $x_{\U}\in\A^{\U}$ the
restriction of $x$ to $\U$.  Let $\U\subset\M$ be a finite subset,
$\gs$ is a {\em subshift of finite type of order $\U$} if there exists
$\mathcal{F}\subset\A^{\U}$ such that $x\in\gs \Longleftrightarrow
x_{m+\U}\in\mathcal{F} \quad \forall m\in\M$. In other word, $\gs$ can
be viewed as a tiling where the allowed patterns are in $\mathcal{F}$.

In the sequel, we will consider \emph{tile sets} and ask whether they
can tile the plane or not. In our formalism, a tile set is a subshift
of finite type: a set of states (the tiles) given together with a set
of allowed patterns (the tiling constraints). We will restrict to
$2\times 1$ and $1\times 2$ patterns (dominos) since it is sufficient
to have the undecidability results of Berger~\cite{berger}.

A {\em cellular automaton} (CA) is a pair $(\am,F)$ where
$F:\am\to\am$ is defined by $F(x)_m=f((x_{m+u})_{u\in\U})$ for all
$x\in\am$ and $m\in\M$ where $\U\subset\Z$ is a finite set named {\em
  neighborhood} and $f:\A^{\U}\rightarrow\A$ is a {\em local
  rule}. The radius of $F$ is $r(F)=\max\{\|u\|_\infty:u\in\U\}$. By
Hedlund's theorem~\cite{hedlund}, it is equivalent to say that $F$ is
a continuous function which commutes with the shift (i.e. $\s^m\circ
F=F\circ\s^m$ for all $m\in\M$).

We recall here general definitions of topological dynamics used all
along the article. Let $(X,d)$ be a metric space and $F:X\to X$ be a
continuous function.

$\bullet$ $x\in X$ is an {\em equicontinuous point} if for all $\e>0$,
there exists $\dd>0$, such that for all $y\in X$, if $d(x,y)<\dd$ then
$d(F^n(x),F^n(y))<\e$ for all $n\in\N$.

$\bullet$ $(X,F)$ is {\em sensitive} if there exists $\e>0$ such that
for all $\dd>0$ and $x\in X$, there exists $y\in X$ and $n\in\N$ such
that $d(x,y)<\dd$ and $d(F^n(x),F^n(y))>\e$. 

\section{Non Sensitive CA Without Any Equicontinuous Point}
\label{sec:sensitivity}

In this section, we will construct a 2D CA which has no equicontinuous
point and is not sensitive to initial conditions. This is in contrast
with dimension 1 where any non-sensitive CA must have equicontinuous
points as shown in~\cite{Kurka97}.

The CA (denoted by $\acf$ in the following) is made of two components:
\begin{itemize}
\item an \emph{obstacle component} (almost static) for which only
  finite type conditions are checked and corrections are made locally
  ;
\item a \emph{particle component} whose overall behaviour is to move
  left and to bypass obstacles.
\end{itemize}

Formally, $\acf$ has a Moore's neighborhood of radius $2$ ($25$
neighbors) and a state set $\A$ with 12 elements :
${\A=\bigl\{U,D,0,1,\oT,\oB,\oR,\oL,\oTR,\oTL,\oBR,\oBL\bigr\}}$ where
the subset ${\obstset=\{1,\oT,\oB,\oR,\oL,\oTR,\oTL,\oBR,\oBL\}}$
corresponds to the obstacle component and $\{U,D,0\}$ to the particle
component.

Let $\obst$ be the subshift of finite type of ${\azz}$ defined by the
set of allowed patterns constituted by all the $3\times 3$ patterns
appearing in the following set of finite configurations:

\[\footnotesize
\begin{matrix}
  \minislot{*} & \minislot{*} & \minislot{*} & \minislot{*} & \minislot{*} & \minislot{*} & \minislot{*} & \minislot{*} & \minislot{*} & \minislot{*} \\
  \minislot{*} & \minislot{*} & \minislot{*} & \minislot{*} & \minislot{*} & \minislot{*} & \minislot{*} & \minislot{*} & \minislot{*} & \minislot{*} \\
  \minislot{*} & \minislot{*} & \minislot{*} & \minislot{\oTL} & \minislot{\oT} & \minislot{\oT} & \minislot{\oT} & \minislot{\oTR} & \minislot{*} & \minislot{*} \\
  \minislot{*} & \minislot{*} & \minislot{*} & \minislot{\oL} & \minislot{1} & \minislot{1} & \minislot{1} & \minislot{\oR} & \minislot{*} & \minislot{*} \\
  \minislot{*} & \minislot{*} & \minislot{*} & \minislot{\oL} & \minislot{1} & \minislot{1} & \minislot{1} & \minislot{\oR} & \minislot{*} & \minislot{*} \\
  \minislot{*} & \minislot{*} & \minislot{*} & \minislot{\oL} & \minislot{1} & \minislot{1} & \minislot{1} & \minislot{\oR} & \minislot{*} & \minislot{*} \\
  \minislot{*} & \minislot{*} & \minislot{*} & \minislot{\oBL} & \minislot{\oB} & \minislot{\oB} & \minislot{\oB} & \minislot{\oBR} & \minislot{*} & \minislot{*} \\
  \minislot{*} & \minislot{*} & \minislot{*} & \minislot{*} & \minislot{*} & \minislot{*} & \minislot{*} & \minislot{*} & \minislot{*} & \minislot{*} \\
  \minislot{*} & \minislot{*} & \minislot{*} & \minislot{*} & \minislot{*} & \minislot{*} & \minislot{*} & \minislot{*} & \minislot{*} & \minislot{*}
\end{matrix}
\]

where $\ast$ stand for any state in $\A\setminus\obstset$.  

In the sequel, a configuration $x$ is said to be \emph{finite} if the
set ${\bigl\{z : x(z)\not=0\bigr\}}$ is finite. Moreover, in such a
configuration, we call \emph{obstacle} a maximal $4$-connected region
of states from $\obstset$.

The following lemma (the proof is straightforward) states that finite
configurations from $\obst$ consist of rectangle obstacles inside a
free $\A\setminus\obstset$ background. Moreover, obstacles are spaced
enough to ensure that any position ``sees'' at most one obstacle in
its $3\times 3$ neighborhood.

\begin{lemma}
  Let ${x\in\obst}$ be a finite configuration. For any $z\in\zz$ we
  have the following:
  \begin{itemize}
  \item either ${x(z)\in\obstset}$ and $z$ belongs to a rectangular obstacle;
  \item or ${x(z)\not\in\obstset}$ and the set of positions
    ${\bigl\{z' : x(z')\in\obstset\text{ and }\|z'-z\|_\infty\leq 1\bigr\}}$
    is empty or belongs to the same obstacle.
  \end{itemize}
\end{lemma}

The local transition function of $\acf$ can be sketched as follows:
\begin{itemize}
\item states from $\obstset$ are turned into $0$'s if finite type
  conditions defining $\obst$ are violated locally and left unchanged
  in any other case ;
\item states $U$ and $D$ behave like a left-moving particle when $U$
  is just above $D$ in a background of $0$'s, and they separate to
  bypass obstacles, $U$ going over and $D$ going under, until they
  meet at the opposite position and recompose a left-moving particle
  (see figure~\ref{fig:partdyn}).
\end{itemize}

\begin{figure}
  \centering
  \includegraphics[width=.4\linewidth]{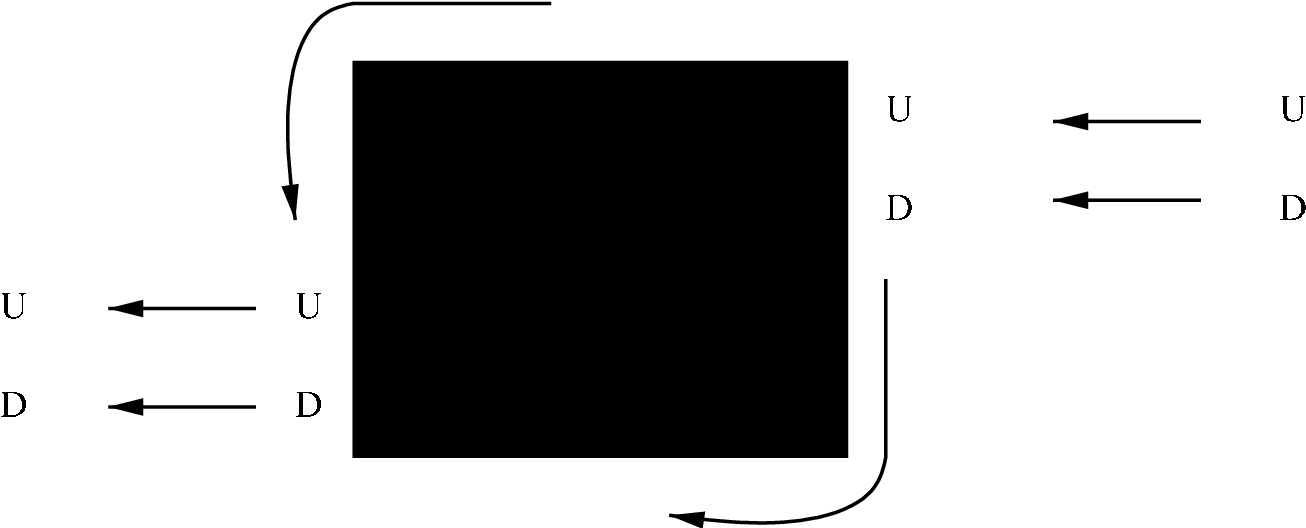}
  \caption{A particle separating into two parts ($U$ and $D$) to
    bypass an obstacle (the black region).}
  \label{fig:partdyn}
\end{figure}

A precise definition of the local transition function of $\acf$ is the following:
\begin{enumerate}
\item if the neighborhood ($5\times 5$ cells) forms a pattern
  forbidden in $\obst$, then turn into state $0$ ;
\item else, apply (if possible) one of the transition rules depending
  only on the ${3\times 3}$ neighborhood detailed in
  figure~\ref{fig:defi}
  \begin{figure}
    \centering
    {\tiny\[\begin{array}{c@{\hspace{.5cm}}c@{\hspace{.5cm}}c@{\hspace{.5cm}}c}
        \transirot{\ast}{\ast}{\ast}{\obstset}{x}{\ast}{\obstset}{\obstset}{\ast}x,&
        \transirot{\obstset}{\obstset}{\ast}{\obstset}{x}{\ast}{\obstset}{\obstset}{\ast}x,&
        \transi{\obstset}{\obstset}{\obstset}{\obstset}{x}{\obstset}{\obstset}{\obstset}{\obstset}x,\\\\
        \transi{0/\obstset}{0}{0}{\obstset}{0}{0}{\obstset}{U}{0}U, &
        \transi{0}{0}{0}{0}{0}{0}{\obstset}{U}{0}U, &
        \transi{0}{0}{0}{0}{0}{U}{\obstset}{\obstset}{D/\obstset}U,
        \\\\ \transi{0}{0}{0}{0}{0}{U}{0}{0/\obstset}{\obstset}U, &
        \transi{0}{U}{0/\obstset}{0}{0}{\obstset}{0}{0}{\obstset}U, & \transi{0}{U}{\obstset}{0}{0}{\obstset}{0}{0}{D}U, \\\\
        \transi{\obstset}{D}{0}{\obstset}{0}{0}{0/\obstset}{0}{0}D, &
        \transi{\obstset}{D}{0}{0}{0}{0}{0}{0}{0}D, &
        \transi{\obstset}{\obstset}{U/\obstset}{0}{0}{D}{0}{0}{0}D,
        \\\\ \transi{0}{0/\obstset}{\obstset}{0}{0}{D}{0}{0}{0}D, &
        \transi{0}{0}{\obstset}{0}{0}{\obstset}{0}{D}{0/\obstset}D,& \transi{0}{0}{U}{0}{0}{\obstset}{0}{D}{\obstset}D\\\\
        \transi{0/\obstset}{0/\obstset}{U}{0/\obstset}{0}{D}{0/\obstset}{0/\obstset}{0/\obstset}D,& \transi{0/\obstset}{0/\obstset}{0/\obstset}{0/\obstset}{0}{U}{0/\obstset}{0/\obstset}{D}U\\
      \end{array}\]}
    \caption{Transition rule of $\acf$ where $x$ stands for any state
      in $\obstset$, '$\ast$' means any state in $\A\setminus\obstset$
      (2 occurrences of $\ast$ are independent), and curved arrows
       mean that the transition is the same for
      any rotation of the neighborhood pattern.}
\label{fig:defi}
\end{figure}

\item in any other case, turn into state $0$.
\end{enumerate}

The possibility to form arbitrarily large obstacles prevents $\acf$ from
being sensitive to initial conditions.

\begin{proposition}
  \label{prop:nosens}
  $\acf$ is not sensitive to initial conditions.
\end{proposition}
\begin{proof}
  Let ${\epsilon>0}$. Let $c_\epsilon$ be the configuration everywhere
  equal to $0$ except in the square region of side
  ${2\bigl\lceil-\log\epsilon\bigr\rceil}$ around the centre where
  there is an obstacle.  ${\forall y\in\azz}$, if
  ${d(y,c_\epsilon)\leq\epsilon/4}$ then ${\forall t\geq 0}$,
  ${d\bigl(\acf^t(c_\epsilon),\acf^t(y)\bigr)\leq\epsilon}$ since a
  well-formed obstacle (precisely, a partial configuration that would
  form a valid obstacle when completed by $0$ everywhere) is
  inalterable for $\acf$ provided it is surrounded by states in
  $\A\setminus\obstset$ (see the 3 first transition rules of case 2 in
  the definition of the local rule): this is guarantied for $y$ by the
  condition ${d(y,c_\epsilon)\leq\epsilon/4}$.\qed
\end{proof}

The next lemma shows that $\obst$ attracts any finite configuration
under the action of $\acf$.

\begin{lemma}
  \label{lem:finiteattrak}
  For any finite configuration $x$, there exists $t_0$ such that
  ${\forall t\geq t_0}$ : ${\acf^t(x)\in\obst}$.
\end{lemma}

The following lemma establishes the key property of the dynamics of
$\acf$: particles can reach any free position inside a finite field of
obstacles from arbitrarily far away from the field.

\begin{lemma}
  \label{lem:partattack}
  Let ${x\in\obst\cap\bigl(\{0\}\cup\obstset\bigr)^\zz}$ be a finite
  configuration. For any $z_0\in\zz$ such that $x(z_0)=0$ there exists
  a path ${(z_n)}$ such that:
  \begin{enumerate}
  \item ${\|z_n\|_\infty\rightarrow\infty}$
  \item ${\exists n_0,\forall n\geq n_0}$, if $x_n$ is the
    configuration obtained from $x$ by adding a particle at position
    $z_n$ (precisely, ${x_n(z_n)=U}$ and
    ${x_n\bigl(z_n+(0,-1)\bigr)=D}$) then
    ${\bigl(\acf^n(x_n)\bigr)(z_0)\in\{U,D\}}$.
  \end{enumerate}
\end{lemma}
\begin{proof}
  First, since $x\in\obst$ and $x(z_0)=0$, then either
  ${x\bigl(z_0+(0,1)\bigr)=0}$ or ${x\bigl(z_0+(0,-1)\bigr)=0}$. We
  will consider only the first case since the proof for the second one
  is similar. Let $(z_n)$ be the path starting from $z_0$ defined as
  follows:
  \begin{itemize}
  \item If $x\bigl(z_n+(1,0)\bigr)=0$ and $x\bigl(z_n+(1,-1)\bigr)=0$ then ${z_{n+1}=z_n+(1,0)}$.
  \item Else, position $z_n+(1,0)$ and/or position $z_n+(1,-1)$
    belongs to an obstacle $P$.  Let $a$, $b$ and $c$ be the positions
    of the upper-left, upper-right and lower-right outside corners of
    $P$ and let $p$ be its half perimeter. Then define
    ${z_{n+1},\ldots,z_{n+p+1}}$ to be the sequence of positions made
    of (see figure~\ref{fig:thepath}):

    \begin{figure}
      \centering
      \includegraphics[scale=.5]{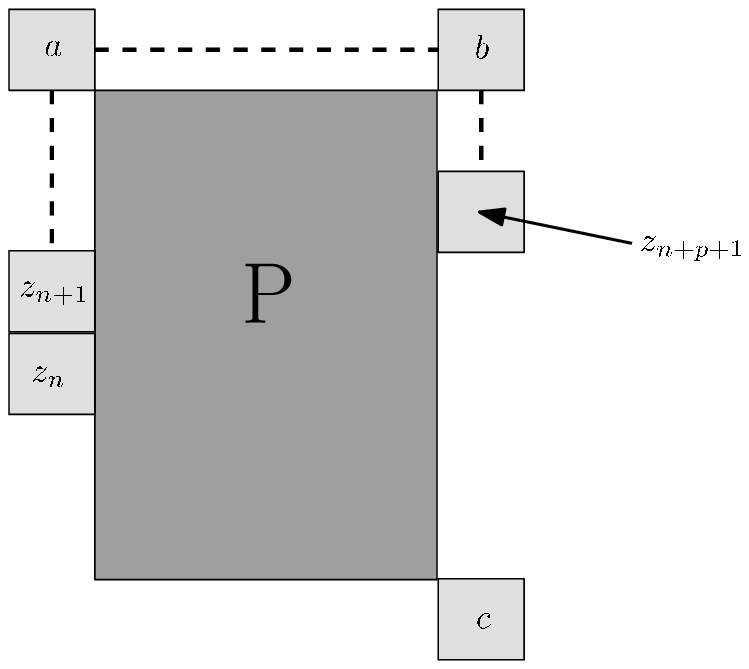}
      \caption{Definition of the $(z_n)$ path in the presence of obstacles.}
      \label{fig:thepath}
    \end{figure}
    \begin{itemize}
    \item a (possibly empty) vertical segment from $z_n$ to $a$,
    \item the segment $[a;b]$,
    \item a (possibly empty) vertical segment from $b$ to $z_{n+p+1}$
      where $z_{n+p+1}$ is the point on $[b;c]$ such that ${z_na +
        bz_{n+p+1}=bc}$.
    \end{itemize}
  \end{itemize}
  We claim that the path $(z_n)$ constructed above has the properties
  of the lemma. Indeed, one can check that for each case of the
  inductive construction of a point $z_m$ from a point $z_n$ we have:
  \begin{itemize}
  \item ${\|z_m\|_\infty>\|z_n\|_\infty}$,
  \item ${\bigl(\acf^{m-n}(x_m)\bigr)(z_n)=U}$ and
    ${\bigl(\acf^{m-n}(x_m)\bigr)(z_n+(0,-1))=D}$ (straightforward from
    the definition of $\acf$).\qed
  \end{itemize}
\end{proof}

\begin{proposition}
  \label{prop:noequ}
  $\acf$ has no equicontinuous points.
\end{proposition}
\begin{proof}
  Assume $\acf$ has an equicontinuous point, precisely a point $x$ which
  verifies ${\forall\epsilon>0,\exists\delta : \forall y,
    d(x,y)\leq\delta\Rightarrow \forall t,
    d\bigl(\acf^t(x),\acf^t(y)\bigr)\leq\epsilon}$.

  Suppose that there is $z_0$ such that ${x(z_0)=0}$ and let
  ${\epsilon = 2^{-\|z_0\|_\infty-1}}$. We will show that the
  hypothesis of $x$ being an equicontinuous point is violated for this
  particular choice of $\epsilon$. Consider any ${\delta>0}$ and let
  $y$ be the configuration everywhere equal to $0$ except in the
  central region of radius ${-\log \left\lceil\delta\right\rceil}$
  where it is identical to $x$. Since $y$ is finite, there exists
  $t_0$ such that ${y_+=\acf^{t_0}(y)\in\obst}$ (by
  lemma~\ref{lem:finiteattrak}).  Moreover, the proof of
  lemma~\ref{lem:finiteattrak} guaranties that for any positive
  integer $t$, ${\bigl(\acf^t(y_+)\bigr)(z_0)=x(z_0)=0}$. So we can apply
  lemma~\ref{lem:partattack} on $y_+$ and position $z_0$ to get the
  existence of a path ${(z_n)}$ allowing particles placed arbitrarily
  far away from $z_0$ to reach the position $z_0$ after a certain
  time. For any sufficiently large $n$, we can construct a
  configuration $y'$ obtained from $y$ by adding a particle at
  position $z_n$. By the property of ${(z_n)}$, we have:
  ${\bigl(\acf^n(y)\bigr)(z_0)\not=\bigl(\acf^n(y')\bigr)(z_0)}$ and
  therefore ${d\bigl(\acf^n(y),\acf^n(y')\bigr)>\epsilon}$. Since, if
  ${n>-\log \left\lceil\delta\right\rceil}$, both $y$ and $y'$ are in
  the ball of centre $x$ and radius $\delta$, we have the desired
  contradiction.

  Assume now that ${\forall z, x(z)\in\obstset}$. There must exist
  some $z_0$ such that ${x(z_0)\not=1}$ (since the uniform
  configuration everywhere equal to 1 is not an equicontinuous point).
  It follows from the definition of $\obst$ that $z_0$ belongs to a
  forbidden pattern for $\obst$. Therefore
  ${\bigl(\acf(x)\bigr)(z_0)=0}$ and we are brought back to the
  previous case of this proof.\qed
\end{proof}

\section{Undecidability of Topological Classification Revisited}
\label{sec:classif}


We will use simulations of Turing machines by tile sets in the
classical way (originally suggested by Wang~\cite{Wang}): the tiling
represents the space-time diagram of the computation and the
transition rule of the Turing machine are converted into tiling
constraints.  Without loss of generality, we only consider Turing
machines working on a semi-infinite tape with a single final state.
The $i^{th}$ machine of this kind in a standard enumeration is denoted
by $\machine{i}$.  In the sequel we use the following notations.
First, to each $\machine{i}$ we associate a tile set $\tileset{i}$
whose constraints ensure the simulation of $\machine{i}$ as mentioned
above; Second, when constructing a CA $G$, we denote by
$\obstacl{G}$ the subshift of its admissible obstacles, which plays
the same role as $\obst$ for $\acf$ with some differences detailed
below.

In \cite{varouch}, the authors give a recursive construction which
produce either a 1D sensitive CA or a 1D CA with equicontinuous points
according to whether a Turing machine halts on the empty input. By
noticing that a 1D CA is sensitive (resp. has equicontinuous points)
in the 1D topology if and only if it is sensitive (resp. has
equicontinuous points) in the 2D topology when viewed as a 2D CA
(neighbors are aligned, e.g. horizontally), we get the following
proposition.

\begin{proposition}
  \label{prop:varouch}
  There is a recursive function ${\phivarouch : \N\rightarrow CA}$
  such that ${\phivarouch(i)\in\equpt}$ if $\machine{i}$ halts on the
  empty input and ${\phivarouch(i)\in\sensi}$ otherwise.
\end{proposition}

However, this is not enough to establish the overall undecidability of
the topological classification of 2D CA. The main concern of this
section is to complete proposition~\ref{prop:varouch} in order to
prove a stronger and more complete undecidability result summarized in
the following theorem.

\begin{theorem}
  \label{theo:undeci}
  Each of the class $\equpt$, $\sensi$ and $\nono$ is neither r.e. nor
  co-r.e.  Moreover any pair of them is recursively inseparable.
\end{theorem}

The proof of this theorem rely on different variants of the
construction of the automaton $\acf$ above. Each time, the construction
scheme is the same, and the desired property is obtained by adding
various contents inside obstacles and slightly changing the rules of
destruction of obstacles according to that content.

The next proposition can be established by using a mechanism to bound
or not the size of admissible obstacles according to whether a Turing
machine halts or not. The idea is to force the tiling representation
of a computation on a blank tape in each obstacle (using the lower
left corner) and to forbid the final state. The proof mechanism used
for $\acf$ can be applied if there is no bound on admissible
obstacles. Otherwise, we get a sensitive CA.

\begin{proposition}
  \label{prop:classdir}
  There is a recursive function ${\phiclassdir : \N\rightarrow CA}$
  such that ${\phiclassdir(i)\in\sensi}$ if $\machine{i}$ halts on the
  empty input and ${\phiclassdir(i)\in\nono}$ otherwise.
\end{proposition}

Before going on with the different constructions needed to prove
theorem~\ref{theo:undeci}, let us stress the dynamical consequence of
the construction of proposition~\ref{prop:classdir}.  It is well-known
that for any 1D sensitive CA of radius $r$, ${2^{-2r}}$ is always the
maximal admissible sensitivity constant (see for
instance~\cite{Kurka97}). Thanks to the above construction it is easy
to construct CA with tiny sensitivity constants as shown by the
following proposition.

\begin{proposition}
  \label{prop:constant}
  The (maximal admissible) sensitivity constant of sensitive 2D CA
  cannot be recursively (lower-)bounded in the number of states and
  the neighborhood size.
\end{proposition}
\begin{proof}
  It is straightforward to check that for each CA $\phiclassdir(i)$
  where $\machine{i}$ halts after $n$ steps on the empty input, the
  maximal admissible obstacle is of height $O(n)$ and of width at
  least $O(\log(n))$. The proposition follows since the sensitivity
  constant of any CA $\phiclassdir(i)\in\sensi$ is precisely ${2^{-l/2
      + 1}}$ where $l$ is the minimum between the largest height and
  the largest width of admissible obstacles.\qed
\end{proof}

Back to the path towards theorem~\ref{theo:undeci}, the following
proposition uses the same ideas as proposition~\ref{prop:classdir} but
it exchanges the role of halting and non-halting computations.

\begin{proposition}
  \label{prop:classinv}
  There is a recursive function ${\phiclassinv : \N\rightarrow CA}$
  such that ${\phiclassinv(i)\in\nono}$ if $\machine{i}$ halts on the
  empty input and ${\phiclassinv(i)\in\sensi}$ otherwise.
\end{proposition}

The properties of the CA $\acf$ and the other constructions above rely
on the fact that obstacles able to stop or deviate particles cannot be
fit together to form larger obstacles. Thus, $\acf$ and other CA have
no equicontinuous point. In the following, we will use a new kind of
obstacles: they are protected from particles by a boundary as the
classical obstacles of $\acf$, but they are made only of successive
boundaries like onion skins. With this new construction it is not
difficult to build an equicontinuous point provided there are
arbitrarily large valid obstacles. The next proposition use this idea
to reduce existence of equicontinuous point to a tiling problem.

\begin{proposition}
  \label{prop:emboit}
  There is a recursive function $\phiemboit$ which associate with any
  tile set $T$ a CA $\phiemboit(T)$ which is in class $\equpt$ if $T$
  tiles the plane and in class $\nono$ otherwise.
\end{proposition}
\begin{proof}[sketch]
  Given a tile set $T$, the CA $\phiemboit(T)$ is identical to $\acf$,
  except that it has a second kind of obstacles, called $T$-obstacles.
  $T$-obstacles are square patterns of states from the set ${E=T\times
    X}$ with ${X=\{\oT,\oB,\oR,\oL,\oTR,\oTL,\oBR,\oBL,\outo\}}$ and
  where the $T$ component is a valid tiling and the $X$ component is
  made from the set of ${2\times 2}$ patterns appearing in the
  following finite configuration:
  \[\tiny
  \begin{matrix}
     &&&  \minislot{\oT} & \minislot{\oT} \\
     &\minislot{\oTL} & \minislot{\oT} & \minislot{\oT} & \minislot{\oT} & \minislot{\oTR} \\
     \minislot{\oL}&\minislot{\oL} & \minislot{\oTL} & \minislot{\oT} & \minislot{\oTR} & \minislot{\oR}& \minislot{\oR} \\
     \minislot{\oL}&\minislot{\oL} & \minislot{\oL} & \minislot{\outo} & \minislot{\oR} & \minislot{\oR}& \minislot{\oR} \\
     &\minislot{\oL} & \minislot{\oBL} & \minislot{\oB} & \minislot{\oBR} & \minislot{\oR}  \\
     &\minislot{\oBL} & \minislot{\oB} & \minislot{\oB} & \minislot{\oB} & \minislot{\oBR}  \\
     &\minislot{\oB} & \minislot{\oB} & \minislot{\oB} & \minislot{\oB} & \minislot{\oB}  \\
     &&&  \minislot{\oB} & \minislot{\oB} \\
  \end{matrix}
  \]
  The $X$ component is used to give everywhere in $T$-obstacles a
  local notion of \emph{inside} and \emph{outside} as depicted by
  figure~\ref{fig:inout} (up to $\pi/2$ rotations): roughly speaking,
  arrows point to the inside region.  Other constraints concerning
  $T$-obstacles are checked locally:
  \begin{itemize}
  \item any pair of obstacles must be at least $2$ cells away from
    each other;
  \item their shape must be a square and this is ensured by requiring
    that any cell of a $T$-obstacle can have ${\{0,U,D\}}$ neighbors
    only in its outside region;
  \item the behaviour of particles with $T$-obstacles is the same as
    with classical obstacles (they can not cross them and states $U$
    and $D$ separate to bypass them).
  \end{itemize}  
  \begin{figure}
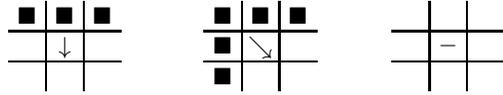

    \centering{\footnotesize
      \begin{tabular}{c|c|c}
        \nslot{\posout} & \nslot{\posout} & \nslot{\posout}\\
        \hline
        \nslot{\posin} & \nslot{\oT} & \nslot{\posin}\\
        \hline
        \nslot{\posin} & \nslot{\posin} & \nslot{\posin}
      \end{tabular} \hskip 1cm
      \begin{tabular}{c|c|c}
        \nslot{\posout} & \nslot{\posout} & \nslot{\posout}\\
        \hline
        \nslot{\posout} & \nslot{\oTL} & \nslot{\posin}\\
        \hline
        \nslot{\posout} & \nslot{\posin} & \nslot{\posin}
      \end{tabular}\hskip 1cm
      \begin{tabular}{c|c|c}
        \nslot{\posin} & \nslot{\posin} & \nslot{\posin}\\
        \hline
        \nslot{\posin} & \nslot{\outo} & \nslot{\posin}\\
        \hline
        \nslot{\posin} & \nslot{\posin} & \nslot{\posin}
      \end{tabular}}
    \caption{Inside (white) and outside (black) positions for states
      of $X$.}
    \label{fig:inout}
  \end{figure}
  The overall dynamics of $\phiemboit(T)$ is similar to that of $\acf$
  with the following exception, which is the key point of the
  construction: destruction of non-valid $T$-obstacles is done
  progressively to preserve as much as possible valid zones inside
  non-valid obstacles.  More precisely any cell in a state from $E$
  remains unchanged unless one of the following
  conditions is verified in which case it turns into state $0$:
  \begin{itemize}
  \item if there is an error in the inside neighborhood;
  \item if there is a position in the outside neighborhood such that
    the pattern formed by that position together with the cell itself
    is forbidden;
  \item if there is a state from $E$ in the $5\times 5$ neighborhood
    which is not connected to the cell by states from $E$.
  \end{itemize}
  One can check that proposition~\ref{prop:nosens} is still true.
  Moreover proposition~\ref{prop:noequ} is true if and only if their
  is a bound on the size of valid $T$-obstacle. Indeed,
  lemma~\ref{lem:finiteattrak} and \ref{lem:partattack} are always
  true and the only point which can be eventually false with the CA
  $\phiemboit(T)$ is the last case in the proof of
  proposition~\ref{prop:noequ}. Precisely, if a valid configuration
  $x$ such that ${\forall z, x(z)\in E}$ can be constructed, then it
  is an equicontinuous point. Otherwise, if such a $x$ is not valid,
  then it contains an error somewhere and the proof scheme of
  proposition~\ref{prop:noequ} can be applied to $\phiemboit(T)$. The
  proposition follows since such a valid $x$ can be constructed if and
  only if $T$ can tile the plane.\qed
\end{proof}

The previous propositions give a set of reductions from
Turing machines or tile sets to CA and one can easily check that the
main theorem follows using Berger's theorem~\cite{berger} and
classical results of the set of halting Turing machines.



To finish this section, we will discuss another difference between 1D
and 2D concerning the complexity of equicontinuous points. Let us
first recall that equicontinuous point in 1D CA can be generated by
finite words often called ``blocking'' words. Precisely, for any
$\acf$ with equicontinuous points, there exists a finite word $u$ such
that ${{}^\infty u^\infty}$ is an equicontinuous point for $\acf$
(proof in~\cite{Kurka97}). The previous construction can be used with
the tile set of Myers~\cite{myers} which can produce only
non-recursive tilings of the plane. Therefore the situation is more
complex in 2D, and we have the following proposition.

\begin{proposition}
  \label{prop:nonrecpoint}
  There exists a 2D CA having equicontinuous points, but only
  nonrecursive ones.
\end{proposition}

\section{Open problems}
\label{sec:open}

It is well-known that equicontinuous CA are exactly ultimately
periodic CA (if they are also bijective, they are periodic).  The
proof techniques developed by Kari in~\cite{kari94} allow to prove
that there is no recursive lower-bound on the pre-period and period of
2D equicontinuous CA. An interesting open question in the continuation
of this paper is to determine whether periods of 1D equicontinuous CA
(bijective or not) can be recursively bounded or not. The only known
result in 1D is that pre-periods are not recursively bounded (this is
essentially the nilpotency problem).

It is interesting to notice that for 1D CA, classes $\sensi$ and
$\equpt$ are easily definable in first-order arithmetic. This is due
to the characterisation by blocking words mentioned above: the
existential quantification over configurations can be replaced by a
quantification over finite words in the definition of $\equpt$.
Proposition~\ref{prop:nonrecpoint} shows that first-order definability
of $\sensi$, $\equpt$ or $\nono$ for 2D CA is more challenging. We
believe they are but at a higher level in the arithmetical hierarchy.

\bibliographystyle{splncs}
\bibliography{ac}

\newpage
\appendix

\section{Proof of lemma~\ref{lem:finiteattrak}}

\begin{proof}
  First, the set ${\bigl\{z : x(z)\in\obstset\}}$ is finite and
  decreasing under the action of $\acf$. Moreover, $U$ and $D$ states
  can only move left, or move vertically or disappear. Since the total
  amount of vertical moves for $U$ and $D$ states is bounded by the
  cardinal of ${\bigl\{z : x(z)\in\obstset\}}$, there is a time $t$
  after which no $U$ or $D$ state is a neighbor of a state of
  $\obstset$, and each $U$ is above a $D$ in a $0$ background (the
  $UD$ particle is on the left of the finite non-$0$ region).  From
  this time on, the evolution of cells in a state of $\obstset$ is
  governed only by the first case of the definition of
  $\acf$. Therefore, after a certain time, finite type conditions
  defining $\obst$ are verified everywhere. To conclude, it is easy to
  check that $\obst$ is stable under the action of $\acf$.\qed
\end{proof}

\section{Proof of proposition~\ref{prop:classdir}}

\begin{proof}[sketch]
  The CA $\phiclassdir(i)$ is constructed from $\acf$ by replacing the
  state $1$ by the tile set $\tileset{i}$ and by defining the obstacle
  subshift $\obstacl{\phiclassdir(i)}$ through the set of all
  ${3\times 3}$ patterns satisfying the following conditions:
  \begin{itemize}
  \item constraints used in the simulation of $\machine{i}$ by $\tileset{i}$ apply;
  \item when replacing states of $\tileset{i}$ by $1$, the resulting
    pattern must be an admissible pattern for $\obst$;
  \item the only allowed state as upper-right neighbor of $\oBL$ is
    $q_0$, the initial state in the simulation of $\machine{i}$ by
    $\tileset{i}$;
  \item the only allowed state above $\oB$ is the blank tape symbol of
    the simulation of $\machine{i}$ by $\tileset{i}$ when $\oBL$ is
    not in the neighborhood;
  \item no final state of the simulation of $\machine{i}$ by
    $\tileset{i}$ is allowed anywhere.
  \end{itemize}
  It is clear from these conditions, that any admissible obstacle of
  $\obstacl{\phiclassdir(i)}$ contains the beginning of the
  computation of $\machine{i}$ on the empty tape.  Since any field of
  obstacles allowed in $\obstacl{\phiclassdir(i)}$ is also allowed in
  $\obst$ (when mapping $\tileset{i}$ to $1$), the proof of
  proposition~\ref{prop:noequ} is still valid and $\phiclassdir(i)$ is
  therefore either in class $\sensi$ or in class $\nono$. Moreover, as
  shown in the proof of proposition~\ref{prop:nosens}, the fact that
  $\phiclassdir(i)$ is sensitive or not depends only on the existence
  of arbitrarily large admissible obstacle. So, by construction,
  $\phiclassdir(i)\in\sensi$ if and only if $\machine{i}$ halts on the
  empty input.\qed
\end{proof}

\section{Proof of proposition \ref{prop:classinv}}

\begin{proof}[sketch]
  Following the proof of proposition~\ref{prop:classdir},
  $\obstacl{\phiclassinv(i)}$ is the same as
  $\obstacl{\phiclassdir(i)}$ with the following differences:
  \begin{itemize}
  \item the final state $q_f$ in the simulation of $\machine{i}$ by
    $\tileset{i}$ is allowed;
  \item the only states allowed at the right and above an occurrence of
    $q_f$ are $q_f$, $\oT$ and $\oR$;
  \item the only state allowed as a lower-left neighbor of $\oTR$ is $q_f$.
  \end{itemize}
  It follows directly that an admissible obstacle must contain the
  final state. Moreover, if $\machine{i}$ halts on the empty input,
  admissible obstacles can be arbitrarily large. The proposition
  follows by a straightforward adaptation of the proofs of
  propositions~\ref{prop:nosens} and \ref{prop:noequ}.\qed
\end{proof}

\end{document}